\documentclass[11pt,a4paper]{article}%
\usepackage{amsfonts}
\usepackage{amsmath}
\usepackage{amssymb}
\usepackage{graphicx}
\usepackage{cancel}
\usepackage{geometry}%
\providecommand{\U}[1]{\protect\rule{.1in}{.1in}}
\newtheorem{theorem}{Theorem}

\newtheorem{corollary}{Corollary}

\newtheorem{lemma}{Lemma}

\newtheorem{proposition}{Proposition}

\newenvironment{proof}[1][Proof]{\textbf{#1.} }{\  \rule{0.5em}{0.5em}}
\makeatletter
\def \@removefromreset#1#2{\let \@tempb \@elt
\def \@tempa#1{@&#1}\expandafter \let \csname @*#1*\endcsname \@tempa
\def \@elt##1{\expandafter \ifx \csname @*##1*\endcsname \@tempa \else
\noexpand \@elt{##1}\fi}     \expandafter \edef \csname cl@#2\endcsname{\csname cl@#2\endcsname}     \let \@elt \@tempb
\expandafter \let \csname @*#1*\endcsname \@undefined}

\@removefromreset{equation}{section}

\@removefromreset{theorem}{section}
\makeatother
\begin{document}

\title{Geometric quantum discord of an arbitrary two-qudit state: the exact value and
general upper bounds}
\author{Elena R. Loubenets$^{1}$ and Louis Hanotel$^{1}$\\$^{1}$Department of Applied Mathematics, HSE University, \\Moscow 101000, Russia}
\maketitle

\begin{abstract}
The geometric quantum discord of a two-qudit state has been studied in many papers, however, its exact analytical value in the explicit form is known only for a general two-qubit state, a general qubit-qudit state and some special families of two-qudit states. Based on the general Bloch vectors formalism [\emph{J. Phys. A: Math. Theor.} 54 195301 (2021)], we find the explicit exact analytical value of the geometric quantum discord \emph{for an arbitrary two-qudit state of any dimension} via the parameters of its correlation matrix and the Bloch vectors of its reduced states. This new
general analytical result includes all the known exact results on the geometric quantum discord only as particular cases and proves rigorously that the lower bound on the geometric discord presented in [\emph{Phys. Rev. A} 85, 024102 (2012)] constitutes its exact value for each two-qudit state.  Moreover, our new general result allows us to
find for an arbitrary two-qudit state, pure or mixed, the novel upper and lower
bounds on its geometric quantum discord, expressed via the Hilbert space
characteristics of this state.
\end{abstract}

\maketitle

\section{Introduction}

As shown by J. Bell theoretically \cite{Bel:64} and later experimentally by
A. Aspect et al. \cite{Asp.Dal.Rog:82}, the probabilistic description of a
quantum correlation scenario does not, in general, agree with the classical
probability model. Nonclassicality of quantum correlations is one of the
main resources for many quantum information processing tasks. Among these
quantum resources, Bell nonlocality and entanglement are the most studied,
see \cite{1, 2, Che.Alb.Fei:05, Lou.Kuz.Han:24} and references therein for
the quantitative and qualitative relations between them.

Nevertheless, there are quantum states that exhibit nonclassical
correlations even without entanglement, and this led to the notion of the
quantum discord \cite{Oll.Zur:01}, which is conceptually rich, however, it
is very hard to calculate it even for a two-qubit state \cite{Gir.Ade:11b}.

Due to the complexity \cite{Hua:14} of computation of the quantum discord,
there were also introduced related concepts, like the measurement-induced
nonlocality \cite{Luo.Fu:11} and the geometric quantum discord \cite%
{Dak.Ved.Bru:10}.

The geometric quantum discord is a geometric measure of quantum correlations
of a bipartite quantum state, which is defined via the distance from this
state to the set $\Omega $ of all states with the vanishing quantum discord 
\cite{Dak.Ved.Bru:10}. In the present article, the geometric quantum discord 
$\mathcal{D}_{g}(\rho )$ of a two-qubit state $\rho $ on $\mathcal{H}%
_{d_{1}}\otimes \mathcal{H}_{d_{2}}$ is defined via the Hilbert-Schmidt norm
between states: 
\begin{equation}
\mathcal{D}_{g}(\rho ):=\min_{\chi \in \Omega }||\rho -\chi ||_{2}^{2}.
\label{vedral}
\end{equation}%
In other definitions \cite{Pau.Oli.Sar:13, Spe.Ors:13,Rog.Spe.Ill:16} of the
geometric quantum discord, different than in (\ref{vedral}) distances are
used.

Though the optimization problem for the computation of the geometric quantum
discord of a bipartite state is much simpler than that for the quantum
discord, its exact value has been explicitly computed only in some
particular cases, namely, for a general two-qubit state \cite{Dak.Ved.Bru:10}%
, a general qubit-qudit state \cite{Ran.Par:12}, a general pure two-qudit state \cite{Luo.Fu:12} and some special families of mixed
two-qudit states \cite{Ran.Par:12, Luo.Fu:10}.

However, to our knowledge, for a general two-qudit state of an arbitrary
dimension, the explicit exact analytical value of the geometric quantum
discord has not been reported in the literature -- only its lower bounds 
\cite{Ran.Par:12, Luo.Fu:10, Has.Lar.Joa:12, Bag.Dey.Osa:19, Gir.Ade:11}.

Geometric quantum discord is a useful concept with applications to quantum state discrimination \cite{Xiong.etal:23}, decoherence \cite{Aar.Lo.Ade:13, Hua.Qiu:16, Zhu.Hu.Li.etal:22, Xio.Bai.etal:23}, quantum phase
estimation, quantum teleportation and remote state preparation protocols, see \cite{Hu.Hu.etal:18} and references therein. For certain states and certain quantum channels, geometric quantum discord has been shown \cite{Hu.Fan:12,Dao.Ahl:12,Mon.Pau.etal:13} to be more resilient than entanglement in dissipative environments, making it a more robust measure for quantifying quantum correlations in decoherence scenarios. Recent studies suggest that geometric quantum discord is also a valuable quantification of quantum correlations in high-energy physics \cite{Wang.Li.etal:23} and quantum gravity \cite{Ban.Bas.etal:23, Fan.Wen.etal:24} contexts. 

In the present paper, for an arbitrary two-qudit state $\rho$ on $\mathcal{H}%
_{d_{1}}\otimes\mathcal{H}_{d_{2}},$ pure or mixed, we find in the explicit
form the exact analytical value of its geometric quantum discord (\ref%
{vedral}). This new rigorously proved general result indicates that the
lower bound on the geometric quantum discord found in \cite{Ran.Par:12} constitutes its exact value for each two-qudit state
and includes only as particular cases the exact results for: (i) general two-qubit \cite{Dak.Ved.Bru:10} and qubit-qudit states \cite%
{Ran.Par:12}; (ii) an arbitrary pure two-qudit state \cite{Luo.Fu:12}; (iii) some special families of two-qudit states \cite%
{Ran.Par:12,Luo.Fu:10}. It also allows us to find the new general upper
bounds on the geometric quantum discord of an arbitrary two-qudit state in
terms of its Hilbert space characteristics.

The paper is organized as follows. In Section 2, we introduce the main
issues of the general Bloch vectors formalism \cite{Lou.Kul:21} for a
finite-dimensional quantum system on which we build up the calculations in
this paper. In Section 3, we find in the explicit analytical form the exact
value of the geometric quantum discord for an arbitrary two-qudit state. In
Section 4, this new result allows us to find new general upper and lower
bounds on the geometric quantum discord in a general two-qudit case. In
Section 5, we discuss the main results of this paper and their importance for the in practical tasks involving two-qudit quantum systems.

\section{Preliminaries: general Bloch vectors formalism}

In this section, we shortly recall the main issues of the general Bloch
vectors mathematical formalism developed in \cite{Lou.Kul:21} for the
description of properties and behavior of a finite-dimensional quantum
system.

Consider the vector space $\mathcal{L}_{d}$ of all linear operators $X$ on a
complex Hilbert space $\mathcal{H}_{d}$ of a finite dimension $d\geq 2.$
Equipped with the scalar product $\langle X_{i},X_{j}\rangle _{\mathcal{L}%
_{d}}:=\mathrm{tr}\left[ X_{i}^{\dag }X_{j}\right] ,$ $\mathcal{L}_{d}$ is a
Hilbert space of the dimension $d^{2},$ referred to as Hilbert-Schmidt.
Denote by 
\begin{align}
\mathfrak{B}_{\Upsilon _{d}}& :=\left\{ \mathbb{I}_{d},\ \Upsilon
_{d}^{(j)}\in \mathcal{L}_{d},\text{\ }\ j=1,\dots,(d^{2}-1)\right\} ,
\label{1_} \\
\Upsilon _{d}^{(j)}& =\left( \Upsilon _{d}^{(j)}\right) ^{\dagger }\neq 0,\
\ \   
 \mathrm{tr[}\Upsilon _{d}^{(j)}]=0,\ \ \ \ \mathrm{tr[}\Upsilon
_{d}^{(j)}\Upsilon _{d}^{(m)}]=2\delta _{jm},  \notag
\end{align}%
an operator basis in $\mathcal{L}_{d}$ consisting of the identity operator $\mathbb{I}%
_{d}$ on $\mathcal{H}_{d}$ and a tuple $\Upsilon _{d}:=\left( \Upsilon
_{d}^{(1)},\dots,\Upsilon _{d}^{(d^{2}-1)}\right) $ of traceless Hermitian
operators mutually orthogonal in $\mathcal{L}_{d}.$ For $d\geq 3,$ some
properties of a particular basis of this type, resulting in \emph{the
generalized Gell-Mann representation}, were considered in \cite%
{Arv.Mal.Muk:97,Jak.Sie:01, Kim:03,
Byr.Kha:03,Sch.Zha.Lea:04,Kim.Kos:05,Men:06,Ber.Kra:08,Lou:20}.

For an arbitrary qudit state $\rho $, the decomposition via a basis (\ref{1_}%
) constitutes the generalized Bloch representation \cite{Lou.Kul:21} 
\begin{align}
\rho _{d}& =\frac{\mathbb{I}_{d}}{d}\ \mathbb{+}\ \sqrt{\frac{d-1}{2d}}%
\left( r_{\Upsilon _{d}}\cdot \Upsilon _{d}\right) ,  \label{2__} \\
r_{\Upsilon _{d}}&\cdot \Upsilon _{d}:=\sum_{j=1}^{d^{2}-1}r_{\Upsilon
_{d}}^{(j)}\Upsilon _{d}^{(j)}, \\
\ r_{\Upsilon _{d}}& =\sqrt{\frac{d}{2(d-1)}}\text{ }\mathrm{tr}\left[ \rho
_{d}\Upsilon _{d}\right] \in \mathbb{R}^{d^{2}-1},  \label{2_1}
\end{align}%
where $r_{\Upsilon _{d}}\in \mathbb{R}^{d^{2}-1}$ is referred to as the
Bloch vector of a qudit state $\rho _{d}$. For a state $\rho_{d},$ the norm
of its Bloch vector satisfies the relations 
\begin{equation}
\left\Vert r_{\Upsilon _{d}}\right\Vert _{\mathbb{R}^{d^{2}-1}}^{2}=\frac{d}{%
d-1}\left( \mathrm{tr}\left[ \rho _{d}^{2}\right] -\frac{1}{d}\right) \leq 1,
\label{3__}
\end{equation}%
and is independent of the choice of a tuple $\Upsilon _{d}$ in an operator
basis (\ref{1_}). For the maximally mixed state, the Bloch vector is equal
to zero.

If a state $\rho _{d}$ is pure, then the norm of its Bloch vector $%
r_{\Upsilon _{d}}$ is equal to $\left\Vert r_{\Upsilon _{d}}\right\Vert _{%
\mathbb{R}^{d^{2}-1}}=1.$ However, in contrast to a qubit case, for an
arbitrary $d>2$, not any unit vector $r\in \mathbb{R}^{d^{2}-1}$ corresponds
via representation 
\begin{equation}
\tau _{d}=\frac{\mathbb{I}_{d}}{d}\mathbb{+}\sqrt{\frac{d-1}{2d}}\left(
r\cdot \Upsilon _{d}\right)   \label{x-2}
\end{equation}%
to a pure state.

Namely, by Proposition 7 and Theorem 2 in \cite{Lou.Kul:21} a Hermitian
operator (\ref{x-2}) with the unit trace constitutes a pure state if and
only if 
\begin{equation}
\left\Vert r\right\Vert _{\mathbb{R}^{d^{2}-1}}^{2}=1,\text{ \ \ }\left\Vert
\left( r\cdot \Upsilon _{d}\right) ^{(-)}\right\Vert _{0}=\sqrt{\frac{2}{%
d(d-1)}},  \label{x-1}
\end{equation}%
where notation $\left\Vert \cdot \right\Vert _{0}$ means the operator norm
of a linear operator on $\mathcal{H}_{d}$ and notation $X^{(-)}$ -- the
nonnegative operator in the unique decomposition of a self-adjoint operator $X$ via $X=X^{(+)}-X^{(-)}$, where $X^{(\pm)}\geq0,\   
X^{(+)}X^{(-)}=X^{(-)}X^{(+)}=0$. 

For $d=2$ and $\Upsilon _{2}=\sigma=(\sigma _{1},\sigma _{2},\sigma _{3})$, where $\sigma$  is the qubit spin operator on $\mathbb{C}^{2}$, the first of relations in (\ref{x-1}) implies the second one.   

From (\ref{2__}) and (\ref{3__}) it follows that, for a state $\rho _{d}$,
the values of the norms $\left\Vert r_{\Upsilon _{d}}\right\Vert _{\mathbb{R}%
^{d^{2}-1}}$ and $\left\Vert \left( r_{\Upsilon _{d}}\cdot \Upsilon
_{d}\right) ^{(-)}\right\Vert _{0}$ do not depend on a choice of a tuple $%
\Upsilon _{d}$ in decomposition (\ref{2__}).

Note that by Lemma 1 in \cite{Lou:20}, the bounds 
\begin{equation}
\sqrt{\frac{2}{d}}\text{ \ }\leq \text{ \ }\frac{\left\Vert r\cdot \Upsilon
_{d}\right\Vert _{0}}{\left\Vert r\right\Vert _{\mathbb{R}^{d^{2}-1}}}\text{
\ }\leq \text{ \ }\sqrt{\frac{2(d-1)}{d}}\   \label{3___}
\end{equation}%
hold for any vector $r\in \mathbb{R}^{d^{2}-1}$ and any tuple $\Upsilon _{d}$%
.

By Eq. (70) in \cite{Lou.Kul:21}, for any two qudit states $\rho _{d},$ $\rho
_{d}^{\prime }$, the scalar product of their Bloch vectors satisfies the
relation 
\begin{equation}
r_{\Upsilon _{d}}\cdot r_{\Upsilon _{d}}^{\prime }\geq -\frac{1}{d-1}\ ,
\label{4_}
\end{equation}%
where equality holds iff $\mathrm{tr}[\rho _{d}\rho _{d}^{\prime }]=0.$

In view of Theorem 2 in \cite{Lou.Kul:21}, relation (\ref{4_}) and identity $%
\sum_{k}|k\rangle \langle k|=\mathbb{I}_{d}$, valid for any orthonormal
basis $\left\{ |k\rangle \in \mathcal{H}_{d},\text{ }k=1,\dots ,d\right\} ,\ 
$we have the following statement needed for our proof of Theorem 1 in
Section 3.

\begin{proposition}
Representation (\ref{x-2}) establishes the one-to-one correspondence between
orthonormal bases $\left\{ |k\rangle \in \mathcal{H}_{d},k=1,\dots
,d\right\} $ in $\mathcal{H}_{d}$ and sets 
\begin{equation}
\Omega _{\Upsilon _{d}}=\left\{ y_{k}\in \mathbb{R}^{d^{2}-1},\text{ }%
k=1,..,d\right\}   \label{z}
\end{equation}%
of vectors $y_{k}$ in $\mathbb{R}^{d^{2}-1},$ satisfying the relations 
\begin{align}
& \sum_{k=1}^{d}y_{k}=0,\ \left\Vert y_{k}\right\Vert _{\mathbb{R}%
^{d^{2}-1}}=1, \ y_{k_{1}}\cdot y_{k_{2}}=-\frac{1}{d-1},  \ \forall
k_{1}\neq k_{2}, \label{z-1} \\
& \left\Vert \left( y_{k}\cdot \Upsilon _{d}\right) ^{(-)}\right\Vert
_{0}=\sqrt{\frac{2}{d(d-1)}},\ \forall k=1,\dots ,d.
\label{z-2}
\end{align}%
\end{proposition}

For a two-qudit state $\rho _{d_{1}\times d_{2}}$ on $\mathcal{\ H}%
_{d_{1}}\otimes \mathcal{H}_{d_{2}},$ $d_{1},d_{2}\geq 2,$ the representation%
\begin{align}
\rho _{d_{1}\times d_{2}}& =\frac{\mathbb{I}_{d_{1}}\otimes \mathbb{I}%
_{d_{2}}}{d_{1}d_{2}}\text{ }+\text{ }\sqrt{\frac{d_{1}-1}{2d_{1}}}\text{ }%
\left( r_{1}\cdot \Upsilon _{d_{1}}\right) \otimes \frac{\mathbb{I}_{d_{2}}}{%
d_{2}}\text{ }  \label{6__} \\
& +\text{ }\sqrt{\frac{d_{2}-1}{2d_{2}}}\text{ }\frac{\mathbb{I}_{d_{1}}}{%
d_{1}}\otimes \left( r_{2}\cdot \Upsilon _{d_{2}}\right)    +\text{ }\frac{1}{4}\sum_{i,j}T_{\rho _{d_{1}\times d_{2}}}^{(ij)}\left(
\Upsilon _{d_{1}}^{(i)}\mathbb{\otimes }\Upsilon _{d_{2}}^{(j)}\right)  
\notag
\end{align}%
is referred \cite{Lou.Kul:21} to as \emph{the generalized Pauli
representation }and constitutes decomposition (\ref{2__}) via the operator basis of
type (\ref{1_}) with the elements having the tensor product form 
\begin{align}
& \mathbb{I}_{d_{1}}\otimes \mathbb{I}_{d_{2}},\text{\ }\Upsilon
_{d_{1}}^{(i)}\otimes \frac{\mathbb{I}_{d_{2}}}{\sqrt{d_{2}}},\ \frac{%
\mathbb{\ I}_{d_{1}}}{\sqrt{d_{1}}}\otimes \Upsilon _{d_{2}}^{(j)},\text{ \ }%
\Upsilon _{d_{1}}^{(i)}\otimes \Upsilon _{d_{2}}^{(j)},   i =1,\dots ,(d_{1}^{2}-1),\text{ \ \ }j=1,\dots ,(d_{2}^{2}-1),  \label{7__} \\
&\Upsilon _{d_{n}}^{(m)} =\left( \Upsilon _{d_{n}}^{(m)}\right) ^{\dagger
}\neq 0,\ \ \mathrm{tr}\left[ \Upsilon _{d_{n}}^{(m)}\right] =0,  \mathrm{tr}\left[ \Upsilon _{d_{n}}^{(m_{1})}\Upsilon _{d_{n}}^{(m_{2})}%
\right] =2\delta _{m_{1}m_{2}},\text{ }n=1,2,  \notag
\end{align}
In representation (\ref{6__}), 
\begin{align}
r_{1}& =\sqrt{\frac{d_{1}}{2(d_{1}-1)}}\mathrm{tr}\left[ \rho _{d_{1}\times
d_{2}}\left( \Upsilon _{d_{1}}\otimes \mathbb{I}_{d_{2}}\right) \right] \in 
\mathbb{R}^{d_{1}^{2}-1},  \label{8__} \\
r_{2}& =\sqrt{\frac{d_{2}}{2(d_{2}-1)}}\mathrm{tr}\left[ \rho _{d_{1}\times
d_{2}}\left( \mathbb{I}_{d_{1}}\otimes \Upsilon _{d_{2}}\right) \right] \in 
\mathbb{R}^{d_{2}^{2}-1},  \label{9__} \\
& \text{ \ }\left\Vert r_{1}\right\Vert _{\mathbb{R}^{d_{1}^{2}-1}}\leq 1,%
\text{ \ }\left\Vert r_{2}\right\Vert _{\mathbb{R}^{d_{2}^{2}-1}}\leq 1, 
\notag
\end{align}%
are the Bloch vectors of states $\rho _{1}=$ \textrm{tr}$_{%
\mathcal{H}_{2}}[\rho _{d_{1}\times d_{2}}]$ and $\rho _{2}=\mathrm{tr}_{%
\mathcal{H}_{1}}[\rho _{d_{1}\times d_{2}}]$ on $\mathcal{H}_{d_{1}}$ and $%
\mathcal{H}_{d_{2}},$ respectively, reduced from a two-qudit state $\rho
_{d_{1}\times d_{2}}$ and satisfying the relation%
\begin{equation}
\mathrm{tr[}\rho _{j}^{2}\mathrm{]}=\frac{1}{d_{j}}+\frac{d_{j}-1}{d_{j}}%
||r_{j}||_{\mathbb{R}^{d_{j}^{2}-1}}^{2},\text{ \ \ }j=1,2,  \label{9-}
\end{equation}%
while 
\begin{equation}
T_{\rho _{d_{1}\times d_{2}}}^{(ij)} :=\mathrm{tr}\left[ \rho _{d_{1}\times
d_{2}}\left( \Upsilon _{d_{1}}^{(i)}\otimes \Upsilon _{d_{2}}^{(j)}\right) %
\right] \ ,  
i =1,\dots ,d_{1}^{2}-1,\text{ \ \ }j=1,\dots ,d_{2}^{2}-1, \label{10__} 
\end{equation}%
are the elements of the real-valued matrix $T_{\rho _{d_{1}\times d_{2}}}$
referred to as the correlation matrix of a two-qudit state $\rho
_{d_{1}\times d_{2}}$. 

In case of a pure two-qudit state $\rho _{d_{1}\times
d_{2}}$,  $%
d_{1},d_{2}\geq 2,$ by the Schmidt theorem \textrm{tr}$[\rho _{1}^{2}]=%
\mathrm{tr}[\rho _{2}^{2}]$ and, in view of relation (\ref{9-}), this implies%
\begin{equation}
\frac{1}{d_{1}}+\frac{d_{1}-1}{d_{1}}\left\Vert r_{1}\right\Vert _{\mathbb{R}%
^{d_{1}^{2}-1}}^{2}=\frac{1}{d_{2}}+\frac{d_{2}-1}{d_{2}}\left\Vert
r_{2}\right\Vert _{\mathbb{R}^{d_{2}^{2}-1}}^{2}.  \label{07_}
\end{equation}
From the generalized Pauli representation (\ref{6__}) it also follows 
\begin{equation}
\mathrm{tr}[\rho _{d_{1}\times d_{2}}^{2}]=\frac{1}{d_{1}d_{2}}+\frac{d_{1}-1%
}{d_{1}d_{2}}\left\Vert r_{1}\right\Vert _{\mathbb{R}^{d_{1}^{2}-1}}^{2} +
\frac{d_{2}-1}{d_{1}d_{2}}\left\Vert r_{2}\right\Vert _{\mathbb{R}%
^{d_{2}^{2}-1}}^{2} +\frac{1}{4}\sum_{i,j}\left( T_{\rho _{d_{1}\times
d_{2}}}^{(ij)}\right) ^{2},  \label{07__}
\end{equation}%

so that expression (\ref{07__}) and relation $\textrm{\ tr}[\rho _{d_{1}\times d_{2}}^{2}]\leq 1$ imply: 
\begin{equation}
 \frac{d_{1}-1}{d_{1}d_{2}}\left\Vert r_{1}\right\Vert _{\mathbb{R}%
^{d_{1}^{2}-1}}^{2}+\frac{d_{2}-1}{d_{1}d_{2}}\left\Vert r_{2}\right\Vert _{%
\mathbb{R}^{d_{2}^{2}-1}}^{2}    +\frac{1}{4}\sum_{i,j}\left( T_{\rho _{d_{1}\times d_{2}}}^{(ij)}\right)
^{2}\leq \frac{d_{1}d_{2}-1}{d_{1}d_{2}},  \label{10.1}
\end{equation}%
where equality holds\ iff a state $\rho _{d_{1}\times d_{2}}$ is pure. 

Since in equality (\ref{07__}) the values of trace $\mathrm{tr}[\rho
_{d_{1}\times d_{2}}^{2}]$ and the Bloch vectors norms $\left\Vert
r_{1}\right\Vert _{\mathbb{R}^{d_{1}^{2}-1}}^{2}$ and $\left\Vert
r_{2}\right\Vert _{\mathbb{R}^{d_{2}^{2}-1}}^{2}$ do not depend on a choice
of tuples $\Upsilon _{d_{1}},\Upsilon _{d_{2}}$ in decomposition (\ref{6__}%
), the same is true for the sum $\sum_{i,j}\left( T_{\rho _{d_{1}\times
d_{2}}}^{(ij)}\right) ^{2}=\mathrm{tr[}T_{\rho _{d_{1}\times
d_{2}}}^{\dagger }T_{\rho _{d_{1}\times d_{2}}}]$, which constitutes the
trace of the positive operator $T_{\rho _{d_{1}\times d_{2}}}^{\dagger
}T_{\rho _{d_{1}\times d_{2}}}$ on $\mathbb{R}^{d_{2}^{2}-1}.$

If $d_{1}=d_{2}=:d,$ then, for every two-qudit state $\rho _{d\times d}$ on $%
\mathcal{\ H}_{d}\otimes \mathcal{H}_{d},$ $d\geq 2,$ 
\begin{equation}
\left\Vert r_{1}\right\Vert _{\mathbb{R}^{d^{2}-1}}^{2}+\left\Vert
r_{2}\right\Vert _{\mathbb{R}^{d^{2}-1}}^{2}+\frac{d^{2}}{4(d-1)}\mathrm{%
tr[}T_{\rho _{d\times d}}^{\dagger }T_{\rho _{d\times d}}]\leq d+1,
\label{10.2}
\end{equation}%
and the bound (48) in \cite{Lou:20} and the above upper bound in (\ref{3___}%
) imply 
\begin{equation}
\left\Vert T_{\rho _{d\times d}}n\right\Vert _{\mathbb{R}^{d^{2}-1}}^{2}
\leq \sqrt{\frac{2}{d}}\sqrt{\frac{2(d-1)}{d}}\left\Vert T_{\rho _{d\times
d}}n\right\Vert _{\mathbb{R}^{d^{2}-1}}\ \Rightarrow \   
\left\Vert T_{\rho _{d\times d}}n\right\Vert _{\mathbb{R}^{d^{2}-1}} \leq 2%
\frac{\sqrt{d-1}}{d},  \label{11.1}
\end{equation}%
for all $n\in \mathbb{R}^{d^{2}-1}\ ,\left\Vert n\right\Vert \leq \frac{1}{%
\sqrt{d-1}}.$ This implies that, for every two-qudit state and any tuples $%
\Upsilon _{d_{1}},\Upsilon _{d_{2}}$, the spectral (operator) norm $%
\left\Vert T\right\Vert _{0}$ of the correlation matrix $T$ is upper bounded
by 
\begin{align}
\left\Vert T_{\rho _{d\times d}}\right\Vert _{0}& :=\sup_{\left\Vert
n\right\Vert =1}\left\Vert T_{\rho _{d\times d}}n\right\Vert _{\mathbb{R}%
^{d^{2}-1}}  \notag \\
& =\sqrt{d-1}\sup_{\left\Vert n\right\Vert =1}\left\Vert T_{\rho _{d\times
d}}\left( \frac{n}{\sqrt{d-1}}\right) \right\Vert _{\mathbb{R}^{d^{2}-1}}
\label{11.2} \\
& \leq \frac{2(d-1)}{d}.  \notag
\end{align}%

Recall that $\left\Vert T_{\rho _{d\times d}}\right\Vert _{0}^{2}$ is the
maximal eigenvalue of the positive self-adjoint operator $T_{\rho _{d\times
d}}^{\dagger }T_{\rho _{d\times d}}.$

Furthermore, for every separable two-qudit state 
\begin{equation}
\rho _{d\times d}^{(sep)}=\sum_{k}\beta _{k}\rho _{1}^{(k)}\otimes \rho
_{2}^{(k)},\text{ \ \ }\beta _{k}\mathcal{\ >}0,\text{ \ \ }\sum_{k}\beta
_{k}\mathcal{=}1,
\end{equation}%
on $\mathcal{H}_{d}\otimes \mathcal{H}_{d},$ $d\geq 2,$ the correlation
matrix $T_{\rho _{d\times d}^{(sep)}}$ and the Bloch vectors (\ref{8__}), (%
\ref{9__}) have the form%
\begin{equation}
T_{\rho _{d\times d}^{(sep)}} =\frac{2(d-1)}{d}\sum_{k}\beta
_{k}|r_{1}^{(k)}\rangle \langle r_{2}^{(k)}|, \  
r_{1} =\sum_{k}\beta _{k}r_{1}^{(k)},\text{ \ \ }r_{2}=\sum_{k}\beta
_{k}r_{2}^{(k)},  \label{11-3}
\end{equation}%
where $r_{j}^{(k)}$ are the Bloch vectors of states $\rho _{j}^{(k)},$ $%
j=1,2,$ given by (\ref{2_1}), and the operator norm of the correlation
matrix is upper bounded by 
\begin{align}
\left\Vert T_{\rho _{d\times d}^{(sep)}}\right\Vert _{0}& \leq \frac{2(d-1)}{%
d}\sum_{k}\beta _{k}\left\Vert |r_{1}^{(k)}\rangle \langle
r_{2}^{(k)}|\right\Vert _{0}  \label{11-4} \\
& =\frac{2(d-1)}{d}\sum_{k}\beta _{k}\left\Vert r_{1}^{(k)}\right\Vert _{%
\mathbb{R}^{d^{2}-1}}\left\Vert r_{2}^{(k)}\right\Vert _{\mathbb{R}%
^{d^{2}-1}}  \notag \\
& \leq \frac{2(d-1)}{d}\ .  \notag
\end{align}%
The concurrence $\mathrm{C}_{|\psi \rangle \langle \psi |}$ of a pure
two-qudit state $\rho _{d_{1}\times d_{2}}=|\psi \rangle \langle \psi |$ on $%
\mathcal{H}_{d_{1}}\otimes \mathcal{H}_{d_{2}},$ $d_{1},d_{2}\geq 2,$ is
defined by the relation 
\begin{equation}
\mathrm{C}_{|\psi \rangle \langle \psi |}=\sqrt{2\left( 1-\mathrm{tr[}\rho
_{j}^{2}]\right) }\ ,  \label{013}
\end{equation}%
and, in view of Eqs. (\ref{9-}), (\ref{07_}), takes the form%
\begin{equation}
\mathrm{C}_{|\psi \rangle \langle \psi |}=\sqrt{2\frac{d_{i}-1}{d_{j}}\left(
1-\left\Vert r_{j}\right\Vert _{\mathbb{R}^{d_{j}^{2}-1}}^{2}\right) },\text{
\ }\ j=1,2.  \label{014}
\end{equation}%
If we introduce the concurrence $\widetilde{\mathrm{C}}_{|\psi \rangle
\langle \psi |}$ normalized to the unity in case of a maximally entangled
quantum state, as it is done in \cite{Lou.Kul:21}, then%
\begin{equation}
\mathrm{C}_{|\psi \rangle \langle \psi |}=\sqrt{2\frac{d_{k}-1}{d_{k}}}%
\widetilde{\mathrm{C}}_{|\psi \rangle \langle \psi |}\text{, \ \ \ \ \ }%
d_{k}\ =\min \{d_{1},d_{2}\},  \label{014.1}
\end{equation}%
and 
\begin{equation}
\widetilde{\mathrm{C}}_{|\psi \rangle \langle \psi |}=\sqrt{1-\left\Vert
r_{k}\right\Vert _{\mathbb{R}^{d_{k}^{2}-1}}^{2}}\ .  \label{015}
\end{equation}%
In a two-qubit case, $\widetilde{\mathrm{C}}_{|\psi \rangle \langle \psi |}=%
\mathrm{C}_{|\psi \rangle \langle \psi |}.$

For a general state $\rho _{d_{1}\times d_{2}},$ pure or mixed,
concurrence $\mathrm{C}_{\rho }$ is defined via the relation%
\begin{equation}
\mathrm{C}_{\rho_{d_{1}\times d_{2}}}=\inf_{\{\alpha _{i},\psi
_{i}\}}\sum \alpha _{i}\mathrm{C}_{|\psi _{i}\rangle \langle \psi _{i}|},
\label{016}
\end{equation}%
where $\rho _{d_{1}\times d_{2}}=\sum_i \alpha _{i}|\psi _{i}\rangle \langle
\psi _{i}|,$ $\sum_i \alpha _{i}=1\ ,$ $\alpha _{i}>0,$ is a possible convex
decomposition of the state $\rho_{d_{1}\times d_{2}}$ via pure states,
see \cite{Che.Alb.Fei:05,Kim.Das.San:09} and references therein.

\vspace*{0.3cm}

\section{Geometric quantum discord}

A general quantum-classical\footnote{%
In this paper, we refer to the right geometric discord instead of left
discord as in \cite{Dak.Ved.Bru:10}.} state on $\mathcal{H}%
_{d_{1}}\otimes \mathcal{H}_{d_{2}}$ has the form  
\begin{align}
\chi _{d_{1}\times d_{2}}& =\sum_{k=1}^{d_{2}}\alpha _{k}\sigma _{k}\otimes
|k\rangle \langle k|,\text{ \ \ }\alpha _{k}\geq 0,\text{ \ \ }%
\sum_{k}\alpha _{k}=1,  \label{11__} \\
|k\rangle & \in \mathcal{H}_{d_{2}},\text{ }\ \text{\ }k=1,\dots ,d_{2},%
\text{ \ }\langle k_{j_{1}}|k_{j_{2}}\rangle =\delta _{j_{1}j_{2}},\text{ \ }
  \ \sum_{k=1}^{d_{2}}\text{ }|k\rangle \langle k|\text{ }=%
\mathbb{I}_{d_{2}}.  \notag
\end{align}%

For short, we further omit the below indices at states indicating its
dimensions at two sites.

In order to find the geometric quantum discord $\mathcal{D}_{g}(\rho
):=\min_{\chi }\mathrm{tr}[(\rho -\chi )^{2}]$ of a state $\rho $, let us
consider the decomposition of the difference between states $\rho $ and $%
\chi $ via their generalized Pauli representations (\ref{6__}). We have 

\begin{align}
\rho-\chi &  =\sqrt{\frac{d_{1}-1}{2d_{1}}}\left(  (r_{1}-\sum_{k=1}^{d_2}\alpha
_{k}x_{k})\cdot \Upsilon_{d_{1}}\right)  \otimes\frac{\mathbb{I}_{{d_{2}}}}{d_{2}}\label{13__}
  +\sqrt{\frac{d_{2}-1}{2d_{2}}}\text{ }\frac{\mathbb{I}_{{d_{1}%
}}}{d_{1}}\otimes\left(  (r_{2}-\sum_{k=1}^{d_2}\alpha_{k}y_{k})\cdot \Upsilon_{d_{2}%
}\right)\\
&  +\frac{1}{4}\sum_{\substack{i=1,\dots,d_{1},\\j=1,\dots,d_{2}}}\left(
T_{\rho}^{(ij)}-2\sqrt{\frac{(d_{1}-1)(d_{2}-1)}{d_{1}d_{2}}}\sum_{k=1}^{d_2}%
\alpha_{k}x_{k}^{(i)}y_{k}^{(j)}\right)  \Upsilon_{d_{1}}^{(i)}\otimes
\Upsilon_{d_{2}}^{(j)},\nonumber
\end{align}
where the Bloch vectors $r_{1}\in \mathbb{R}^{d_{1}^{2}-1}$
and $r_{2}\in \mathbb{R}^{d_{2}^{2}-1}$ and $T_{\rho }^{(ij)}$ are defined
in (\ref{8__}), (\ref{9__}) and (\ref{10__}), respectively, and have norms $%
\left\Vert r_{1}\right\Vert _{\mathbb{R}^{d_{1}^{2}-1}}\leq 1,\left\Vert
r_{2}\right\Vert _{\mathbb{R}^{d_{2}^{2}-1}}\leq 1,$ whereas for $k=1,\dots
,d_{2}$, 

\begin{align}
x_{k}& =\sqrt{\frac{d_{1}}{2\left( d_{1}-1\right) }}\mathrm{tr}\left[ \sigma
_{k}\Upsilon _{d_{1}}\right] \in \mathbb{R}^{d_{1}^{2}-1},\text{\ \ }%
\left\Vert x_{k}\right\Vert _{\mathbb{R}^{d_{1}^{2}-1}}\leq 1,  \label{14__}
\\
y_{k}& =\sqrt{\frac{d_{2}}{2\left( d_{2}-1\right) }}\langle k|\Upsilon
_{d_{2}}|k\rangle \in \mathbb{R}^{d_{2}^{2}-1},\text{ \ \ }\left\Vert
y_{k}\right\Vert _{\mathbb{R}^{d_{2}^{2}-1}}=1,  \notag
\end{align}%
are, respectively, the Bloch vectors of states $\sigma _{k}$ on $%
\mathcal{\ H}_{d_{1}}$ and mutually orthogonal pure states $|k\rangle
\langle k|,$ $\sum_{k=1}^{d_{2}}|k\rangle \langle k|=\mathbb{I}_{d_{2}},$ on 
$\mathcal{\ H}_{d_{2}}.$

By Proposition 1, representation (\ref{x-2}) establishes the one-to-one
correspondence between orthonormal bases in $\mathcal{H}%
_{d_{2}}$ and sets $\Omega _{\Upsilon _{d_{2}}}=\left\{
y_{k}\in \mathbb{R}^{d_{2}^{2}-1},\text{ }k=1,\dots,d_{2}\right\} $ of vectors
in $\mathbb{R}^{d_{2}^{2}-1}$, satisfying the relations: 
\begin{align}
& \sum_{k=1}^{d_{2}}y_{k}=0,\ \left\Vert y_{k}\right\Vert _{\mathbb{\ R}%
^{d_{2}^{2}-1}}=1,\ y_{k_{1}}\cdot y_{k_{2}}=-\frac{1}{d_{2}-1},   \ \forall
k_{1}\neq k_{2}, \label{15__} \\
& \left\Vert \left( y_{k}\cdot \Upsilon _{d_{2}}\right) ^{(-)}\right\Vert
_{0}=\sqrt{\frac{2}{d_{2}(d_{2}-1)}},\ \forall k=1,\dots ,d_{2} \ .
\label{15__1}
\end{align}%
Eq. (\ref{13__}) implies 
\begin{align}
& \mathrm{tr}[\left( \rho -\chi \right) ^{2}]=\frac{d_{1}-1}{d_{1}d_{2}}%
\left\Vert r_{1}-\sum_{k=1}^{d_{2}}\alpha _{k}x_{k}\right\Vert _{\mathbb{R}%
^{d_{1}^{2}-1}}^{2}  +\frac{d_{2}-1}{d_{1}d_{2}}\left\Vert r_{2}-\sum_{k=1}^{d_{2}}\alpha
_{k}y_{k}\right\Vert^{2}_{\mathbb{R}^{d_{2}^{2}-1}}  \label{16__} \\
& +\frac{1}{4}\sum_{i,j}\left( T_{\rho }^{(ij)}-2\sqrt{\frac{%
(d_{1}-1)(d_{2}-1)}{d_{1}d_{2}}}\sum_{k=1}^{d_{2}}\alpha
_{k}x_{k}^{(i)}y_{k}^{(j)}\right) ^{2}  \notag
\end{align}%
and, under conditions (\ref{15__}), relation (\ref{16__}) reduces to 
\begin{align}
& \mathrm{tr}[(\rho -\chi )^{2}]=\frac{d_{2}-1}{d_{1}d_{2}}\left\Vert
r_{2}\right\Vert _{\mathbb{R}^{d^{2}-1}}^{2}+\frac{1}{4}\mathrm{tr}[T_{\rho
}^{\dagger }T_{\rho }]  \label{16.1} \\
& +\frac{d_{1}-1}{d_{1}}\sum_{k=1}^{d_{2}}\left\Vert \alpha _{k}x_{k}-\frac{%
r_{1}}{d_{2}}-\frac{1}{2}\sqrt{\frac{d_{1}\left( d_{2}-1\right) }{%
d_{2}\left( d_{1}-1\right) }}T_{\rho }y_{k}\right\Vert _{\mathbb{R}%
^{d_{1}^{2}-1}}^{2}  \notag \\
& +\frac{1}{d_{1}}\sum_{k=1}^{d_{2}}\left[ \alpha _{k}-\frac{1}{d_{2}}-\frac{%
d_{2}-1}{d_{2}}(r_{2}\cdot y_{k})\right] ^{2}  \notag \\
& -\frac{(d_{2}-1)^{2}}{d_{2}^{2}d_{1}}\sum_{k=1}^{d_{2}}(r_{2}\cdot
y_{k})^{2}-\frac{d_{2}-1}{4d_{2}}\sum_{k=1}^{d_{2}}\left\Vert T_{\rho
}y_{k}\right\Vert _{\mathbb{R}^{d_{1}^{2}-1}}^{2}\ .  \notag
\end{align}

From relation (\ref{16.1}) it follows that, for a fixed set $\{y_{k}\}$ the
minimum of $\mathrm{tr}[(\rho -\chi )^{2}]$ over $x_{k}$ and $\alpha _{k}$
is attained at 
\begin{align}
\alpha _{k}x_{k}& =\frac{r_{1}}{d_{2}}-\frac{1}{2}\sqrt{\frac{d_{1}\left(
d_{2}-1\right) }{d_{2}\left( d_{1}-1\right) }}T_{\rho }y_{k},  \label{16.2} \\
\alpha _{k}& =\frac{1}{d_{2}}+\frac{d_{2}-1}{d_{2}}(r_{2}\cdot
y_{k})\Rightarrow \sum_{k=1}^{d_{2}}\alpha _{k}=1,  \notag
\end{align}%
such that%
\begin{equation}
\left\Vert \frac{r_{1}}{d_{2}}-\frac{1}{2}\sqrt{\frac{d_{1}\left(
d_{2}-1\right) }{d_{2}\left( d_{1}-1\right) }}T_{\rho }y_{k}\right\Vert _{%
\mathbb{R}^{d_{1}^{2}-1}}\leq 1.
\end{equation}%
Taking this into account in relation (\ref{16.1}), we come to the following
statement.
\begin{proposition}
For every two-qudit state $\rho $ on $\mathcal{H}_{d_{1}}\otimes \mathcal{H}%
_{d_{2}}$, the geometric quantum discord $\mathcal{D}_{g}(\rho )=\min_{\chi }%
\mathrm{tr}[(\rho -\chi )^{2}]$ is given by%
\begin{equation}
\mathcal{D}_{g}(\rho ) =\frac{d_{2}-1}{d_{1}d_{2}}\left\Vert
r_{2}\right\Vert _{\mathbb{R}^{d_{2}^{2}-1}}^{2}+\frac{1}{4}\mathrm{tr}%
[T_{\rho }^{\dagger }T_{\rho }]  -\max_{\Omega _{\Upsilon _{d_{2}}}}\mathrm{tr}\left[ \left( \frac{d_{2}-1}{%
d_{1}d_{2}}|r_{2}\rangle \langle r_{2}|\text{ }+\text{ }\frac{1}{4}T_{\rho
}^{\dagger }T_{\rho }\right) \Pi _{\Omega _{\Upsilon _{d_{2}}}}\right] , 
\label{17__}
\end{equation}%
where: (i) $T_{\rho }$ is the correlation matrix (\ref{10__}) of a state $%
\rho $ and $r_{2}$ is the Bloch vector (\ref{9__}) of the reduced state $%
\rho _{2}=\mathrm{tr}_{\mathcal{H}_{1}}[\rho ]$ on $\mathcal{H}_{d_{2}}$
within decomposition (\ref{6__}) specified with arbitrary tuples $\Upsilon
_{d_{1}}$ and $\Upsilon _{d_{2}};$ (ii) the positive Hermitian operator $\Pi
_{\Omega _{\Upsilon _{d_{2}}}}$ on $\mathbb{R}^{d_{2}^{2}-1}$ is defined by
the relation 
\begin{equation}
\Pi _{\Omega _{\Upsilon _{d_{2}}}}:=\frac{d_{2}-1}{d_{2}}%
\sum_{k=1}^{d_{2}}|y_{k}\rangle \langle y_{k}|,  \label{18__}
\end{equation}%
where%
\begin{equation}
\Omega _{\Upsilon _{d_{2}}}=\left\{ y_{k}\in \mathbb{R}^{d_{2}^{2}-1},\text{ 
}k=1,..,d_{2}\right\} \subset \mathbb{R}^{d_{2}^{2}-1}  \label{19___}
\end{equation}%
is a set of linear dependent vectors in $\mathbb{R}^{d_{2}^{2}-1},$
satisfying relations (\ref{15__}), (\ref{15__1}). In (\ref{17__}), notations 
$|r_{2}\rangle $ and $\langle r_{2}|$ mean, the column vector and the line
vector, corresponding to  tuple $r_{2}=(r_{2}^{(1)},\dots
,r_{2}^{(d_{2}^{2}-1)})\in \mathbb{R}^{d_{2}^{2}-1}$.
\end{proposition}

The following statement is proved in Appendix A.

\begin{lemma}
For any tuple $\Upsilon _{d_{2}},$ a positive operator (\ref{18__}) on $\mathbb{R}^{d_{2}^{2}-1}$ is an
orthogonal projection of rank $(d_{2}-1).$
\end{lemma}

Taking into account Proposition 2 and Lemma 1, we proceed to introduce for a two-qudit state $\rho $, pure or mixed and 
of any dimension, the explicit exact analytical value of its    
geometric quantum discord $\mathcal{D}_{g}(\rho )$ in terms of 
characteristics of this state within the generalized Pauli representation (\ref{6__}). 

\begin{theorem}
For an arbitrary two-qudit state $\rho $ on $\mathcal{H}_{d_{1}}\otimes 
\mathcal{H}_{d_{2}},$ $d_{1},d_{2}\geq 2,$ the geometric quantum discord
equals to%
\begin{equation}
\mathcal{D}_{g}(\rho ) =\frac{d_{2}-1}{d_{1}d_{2}}\left\Vert
r_{2}\right\Vert _{\mathbb{R}^{d_{2}^{2}-1}}^{2}+\frac{1}{4}\mathrm{tr}%
[T_{\rho }^{\dagger }T_{\rho }]-\sum_{n=1}^{d_{2}-1}\eta _{n} 
=\sum_{n=d_{2}}^{d_{2}^{2}-1}\eta _{n},\label{25__} 
\end{equation}%
where $\eta _{1}\geq \eta _{2}\geq \ldots \geq \eta _{d_{2}^{2}-1}\geq 0$
are the eigenvalues of the positive Hermitian operator 
\begin{equation}
G(\rho )=\frac{d_{2}-1}{d_{1}d_{2}}|r_{2}\rangle \langle r_{2}|\text{ }+%
\text{ }\frac{1}{4}T_{\rho }^{\dagger }T_{\rho }  \label{26}
\end{equation}%
on $\mathbb{R}^{d_{2}^{2}-1}$  listed in the decreasing
order with the corresponding algebraic multiplicities. The eigenvalues of the positive operator $
G(\rho )$ and the values of the norm $\left\Vert r_{2}\right\Vert _{\mathbb{R}^{d_{2}^{2}-1}}^{2}$ and the
trace $\mathrm{tr}[T_{\rho }^{\dagger }T_{\rho }]$ are independent on a choice of tuples $\Upsilon _{d_{1}}$ and $%
\Upsilon_{d_{2}}$ within representation (\ref{6__}).
\end{theorem}

\begin{proof}
Let the Bloch vector $r_{2}\in \mathbb{R}^{d_{2}^{2}-1}$ and the correlation
matrix $T_{\rho }$  in (\ref{17__}) be defined within decomposition (\ref{6__}%
) for some arbitrary tuples $\Upsilon _{d_{1}}$ and $\Upsilon _{d_{2}}.$ As
indicated in Section 2, the values of the norm $\left\Vert r_{2}\right\Vert
_{\mathbb{R}^{d_{2}^{2}-1}}^{2}\leq 1$ and the trace $\mathrm{tr}[T_{\rho
}^{\dagger }T_{\rho }]$ do not depend on a choice of tuples $\Upsilon
_{d_{1}}$ and $\Upsilon _{d_{2}}$ and are determined only by a state $\rho $%
. By Lemma 1 every positive operator $\Pi _{\Omega _{\Upsilon _{d_{2}}}}$,
given by (\ref{18__}), is an orthogonal projection of rank $(d_{2}-1)$, so
that it has eigenvalue $1$ of multiplicity $(d_{2}-1)$ and the eigenvalue $0$
with multiplicity $d_{2}(d_{2}-1).$ This and the von Neumann inequality \cite%
{Mir:75} $\mathrm{tr}[AB]\leq \sum a_{i}b_{i}$, which is valid for any two positive
operators $A$ and $B,$ with eigenvalues $a_{i}\geq 0$ and $b_{i}\geq 0$,
listed in decreasing order, imply that, for each $\Pi _{\Omega _{\Upsilon
_{d_{2}}}}$, the trace $\mathrm{tr}[G(\rho )\Pi _{\Omega _{\Upsilon
_{d_{2}}}}]$,  standing under the maximum in (\ref{17__}),  is upper bounded by
\begin{equation}
\mathrm{tr}[G(\rho )\Pi _{\Omega _{\Upsilon _{d_{2}}}}]\leq
\sum_{k=1}^{d_{2}-1}\eta _{k}.  \label{0}
\end{equation}%
Therefore, in order to prove the exact analytical expression (\ref{25__}), we have to present some projection $\Pi _{\Omega _{\Upsilon
_{d_{2}}}}^{\prime }$ on which the upper bound (\ref{0}) is attained.

Let us introduce projections $\Pi _{\Omega _{\widetilde{\Upsilon }_{d_{2}}}}=\frac{%
d_{2}-1}{d_{2}}\sum_{k=1}^{d_{2}}|y_{k}\rangle \langle y_{k}|$, which are defined via
vectors $y_{k}\in \Omega _{\widetilde{\Upsilon }_{d_{2}}}$ satisfying
relations (\ref{15__}) and the  condition (\ref{15__1}), specified for some tuple $%
\widetilde{\Upsilon }_{d_{2}}\neq \Upsilon _{d_{2}},$ which we choose below.
Let 
\begin{equation}
\widetilde{\Upsilon }_{d_{2}}^{(j)}=\mathrm{v}_{j}\cdot \Upsilon _{d_{2}},%
\text{ \ }j=1,...,(d_{2}^{2}-1),  \label{new}
\end{equation}%
be the decomposition of traceless hermitian operators $\widetilde{\Upsilon }%
_{d_{2}}^{(j)}$ -- elements of a tuple $\widetilde{\Upsilon }_{d_{2}}$ -- in
the basis (\ref{1_}) specified with tuple $\Upsilon _{d_{2}}.$ The set of
vectors $\{\mathrm{v}_{j}\in \mathbb{R}^{d_{2}^{2}-1},$ $j=1,...,(d_{2}^{2}-1)\}$
constitutes an orthonormal basis, see Eq. 18 in \cite{Lou.Kul:21}.

For a given projection 
\begin{equation}
\widetilde{\Pi }_{\widetilde{\Omega} _{\widetilde{\Upsilon }_{d_{2}}}}=\frac{d_{2}-1}{%
d_{2}}\sum_{k=1}^{d_{2}}|\widetilde{y}_{k}\rangle \langle \widetilde{y}_{k}|,%
\text{ \ \ }k=1,...,d_{2},\text{\ }  \label{_1}
\end{equation}%
denote by $\widetilde{g}_{m},$ $m=1,\dots ,d_2^{2}-1,$ its mutually orthogonal
eigenvectors, where the first $(d_{2}-1)$ eigenvectors correspond to
eigenvalue $1$ and others to the eigenvalue $0$. By the spectral theorem, we
have 
\begin{equation}
G(\rho )=\sum_{n=1}^{d_{2}^{2}-1}\eta _{n}|e_{n}\rangle \langle e_{n}|,\ \ \ 
\widetilde{\Pi }_{\widetilde{\Omega}_{\widetilde{\Upsilon }%
_{d_{2}}}}=\sum_{m=1}^{d_{2}-1}|\widetilde{g}_{m}\rangle \langle \widetilde{g%
}_{m}|.  \label{_2}
\end{equation}%
To projection (\ref{_1}) define via the unitary operator $\mathrm{U}%
=\sum_{n}|e_{n}\rangle \langle \widetilde{g}_{n}|$ the projection 
\begin{equation}
\mathrm{U}\widetilde{\Pi }_{\widetilde{\Omega} _{\widetilde{\Upsilon }_{d_{2}}}}\mathrm{U%
}^{\dagger }=\frac{d_{2}-1}{d_{2}}\sum_{k=1}^{d_{2}}|y'_{k}\rangle \langle
y'_{k}|=\sum_{n=1}^{d_{2}-1}|e_{n}\rangle \langle e_{n}|,  \label{_3}
\end{equation}%
where vectors $y'_{k}=\mathrm{U}\widetilde{y}_{k}\in \mathbb{R}^{d_{2}^{2}-1}$
satisfy relations (\ref{15__}) and also relation (\ref{15__1}) 
\begin{equation}
\sqrt{\frac{2}{d_{2}(d_{2}-1)}} =\left\Vert \left( \widetilde{y}_{k}\cdot 
\widetilde{\Upsilon }_{d_{2}}\right) ^{(-)}\right\Vert _{0}  =\left\Vert \left( y'_{k}\cdot \widetilde{\Upsilon }_{d_{2}}^{\prime
}\right) ^{(-)}\right\Vert _{0},  \label{_4}
\end{equation}%
but with respect to  tuple $\widetilde{\Upsilon }_{d_{2}}^{\prime }$ with
elements%
\begin{equation}
\left( \widetilde{\Upsilon }_{d_{2}}^{\prime }\right) ^{(m)}=\sum_{l}\mathrm{%
U}_{lm}^{\dagger }\widetilde{\Upsilon }_{d_{2}}^{(l)}.  \label{_5}
\end{equation}%
Substituting (\ref{new}) into (\ref{_5}), we derive%
\begin{equation}
\left( \widetilde{\Upsilon }_{d_{2}}^{\prime }\right) ^{(m)}=\sum_{j}\left(
\sum_{l}\mathrm{U}_{lm}^{\dagger }\mathrm{v}_{l}^{(j)}\right) \Upsilon
_{d_{2}}^{(j)}.  \label{_5.1}
\end{equation}%
Choosing in (\ref{new}) the orthonormal basis $\left\{ \mathrm{v}_{l}\in 
\mathbb{R}^{d_{2}^{2}-1},\text{ }l=1,...,(d_2^{2}-1)\right\} $ with components 
$\mathrm{v}_{l}^{(j)}=\mathrm{U}_{jl},$ we have%
\begin{equation}
\sum_{l}\mathrm{U}_{lm}^{\dagger }\mathrm{v}_{l}^{(j)}=\sum_{l}\mathrm{U}%
_{jl}\mathrm{U}_{lm}^{\dagger }=\delta _{jm}\text{}\Rightarrow \text{ \
\ }\widetilde{\Upsilon }_{d_{2}}^{\prime }=\Upsilon _{d_{2}}.  \label{_5.2}
\end{equation}%
Therefore, under the above unitary transform of projection (\ref{_1}) and the
specific choice via (\ref{new}) of a tuple $\widetilde{\Upsilon }_{d_{2}}$ in
(\ref{_1}), we come to the projection 
\begin{equation}
\mathrm{U}\text{ }\widetilde{\Pi }_{\Omega _{\widetilde{\Upsilon }_{d_{2}}}}%
\mathrm{U}^{\dagger }=\frac{d_{2}-1}{d_{2}}\sum_{k=1}^{d_{2}}|y_{k}\rangle
\langle y_{k}|=\Pi _{\Omega _{\Upsilon _{d_{2}}}}^{^{\prime }},  \label{_6}
\end{equation}%
which is included into the set of projections over which the maximum in (\ref%
{17__}) is considered. Taking into account that by relation (\ref{_3}), $\Pi
_{\Omega _{\Upsilon _{d_{2}}}}^{^{\prime
}}=\sum_{n=1}^{d_{2}-1}|e_{n}\rangle \langle e_{n}|$ we have\footnote{%
For our further consideration in Section 4, based on the proof of Theorem 1,
we also formulate in Proposition 4 of Appendix A the general statement on $%
\max_{\Omega _{y}}\mathrm{tr}[A\Pi _{\Omega _{y}}]$ for any positive
Hermitian operator $A$ on $\mathbb{R}^{d^{2}-1}.$}  
\begin{equation}
\mathrm{tr}[G(\rho )\Pi _{\Omega _{\Upsilon _{d_{2}}}}^{^{\prime }}] =%
\mathrm{tr}\left[ \sum_{n=1}^{d_{2}^{2}-1}\eta _{n}|e_{n}\rangle \langle
e_{n}|\text{ }\sum_{m=1}^{d_{2}-1}|e_{m}\rangle \langle e_{m}|\right]  
 =\sum_{n=1}^{d_{2}-1}\eta _{n}.  \label{_7}
\end{equation}%
Eqs. (\ref{0}), (\ref{_7}) prove the statement.
\end{proof}

The new general exact result (\ref{25__}) proved by Theorem 1 indicates that
the lower bound on the geometric quantum discord presented in \cite%
{Ran.Par:12} is attained on every two-qudit state. Moreover, this new exact 
result on the geometric quantum discord includes only as particular cases
all the exact expressions known \cite%
{Dak.Ved.Bru:10, Ran.Par:12} for some particular mixed states.

We note that, in contrast to the derivation of the lower bound in \cite%
{Ran.Par:12} via the Pauli decomposition with the generalized Gell-Mann
operators, our derivation of the exact value (\ref{25__}) is based on the
Pauli decomposition with respect to any operator basis of the from (\ref{7__}%
). Also, the normalization coefficients in (\ref{25__}) are different from
those in \cite{Ran.Par:12} and satisfy the general relations derived in \cite%
{Lou.Kul:21} and presented in short in Section 2.

The following statement shows that, in case of a pure two-qudit state, the
new exact general result (\ref{25__}) in Theorem 1 includes as a particular case the expression \cite{ Luo.Fu:12} for the geometric quantum discord of a pure two-qudit state, which was derived in \cite{ Luo.Fu:12} directly from the definition (\ref{vedral}).

\begin{corollary}
For every pure two-qudit state $\rho _{\psi }=|\psi \rangle \langle \psi |$
on $\mathcal{H}_{d}\otimes \mathcal{H}_{d},$ $d\geq 2,$ the geometric
quantum discord is given by%
\begin{equation}
\mathcal{D}_{g}(\rho _{\psi })=\frac{1}{2}\mathrm{C}_{\rho _{\psi }}^{2}\leq
2\mathcal{N}^{2}_{\rho _{\psi }},  \label{26_1}
\end{equation}%
where the equality in the right-hand side holds for a pure two-qubit state. Here, $\mathrm{C}_{\rho _{\psi }}$ is the concurrence (\ref{013}) of a pure
two-qudit state $\rho _{\psi }$ and $\mathcal{N}_{\rho _{\psi }}$ is its
negativity\footnote{For a pure two-qudit state, the negativity takes the form  $\sum_{1\leq k<m\leq d}\sqrt{\mu _{k}\mu _{m}}$, see, for example, in Section 4 of \cite{Lou.Nam:22}.}. 
For a maximally entangled pure two-qudit state $\rho _{\psi _{\max }}$,%
\begin{equation}
\mathcal{D}_{g}(\rho _{\psi _{\max }})=\frac{d-1}{d}.  \label{26__2}
\end{equation}
\end{corollary}

\begin{proof}
Consider first the geometric quantum discord for a pure two-qubit state. As
it is found in Theorem 2 of \cite{Lou.Kuz.Han:24}, for a pure two-qubit
state, the eigenvalues of the positive operator $T_{\rho _{\psi }}^{\dagger
}T_{\rho _{\psi }}$ are equal to $1,$ $\mathrm{C}_{\rho _{\psi }}^{2}$, $%
\mathrm{C}_{\rho _{\psi }}^{2}$ and if $\left\Vert r_{2}\right\Vert _{%
\mathbb{R}^{3}}^{2}=1-\mathrm{C}_{\rho _{\psi }}^{2}\neq 0$ (that is, a pure
state $|\psi \rangle $ is not maximally entangled), then the Bloch vector $%
r_{2}\in \mathbb{R}^{3}$ constitutes \cite{Lou.Kuz.Han:24} the eigenvector
of matrix $T_{\rho _{\psi }}^{\dagger }T_{\rho _{\psi }}.$ Therefore, if a
pure two-qubit state $|\psi \rangle $ is not maximally entangled, then, in
view of the spectral theorem, the positive operator $G(\rho _{\psi })$ in (%
\ref{26}) takes the form%
\begin{align}
& G(\rho_{\psi } )=\frac{1}{4}|r_{2}\rangle \langle r_{2}|+\frac{1}{4}\frac{%
|r_{2}\rangle \langle r_{2}|}{\left\Vert r_{2}\right\Vert _{\mathbb{R}%
^{3}}^{2}}+\frac{1}{4}\mathrm{C}_{\rho_{\psi } }^{2}\left( |\mathrm{v}_{1}\rangle
\langle \mathrm{v}_{1}|+|\mathrm{v}_{2}\rangle \langle \mathrm{v}%
_{2}|\right)  \notag  \\
& =\frac{1}{4}\left( \left\Vert r_{2}\right\Vert _{\mathbb{R}%
^{3}}^{2}+1\right) \frac{|r_{2}\rangle \langle r_{2}|}{\left\Vert
r_{2}\right\Vert _{\mathbb{R}^{3}}^{2}}+\frac{1}{4}\mathrm{C}_{\rho_{\psi }
}^{2}\left( |\mathrm{v}_{1}\rangle \langle \mathrm{v}_{1}|+|\mathrm{v}%
_{2}\rangle \langle \mathrm{v}_{2}|\right) \ , \label{26_2}  
\end{align}%
where $|\mathrm{v}_{1}\rangle ,$ $|\mathrm{v}_{2}\rangle $ are two mutually
orthogonal eigenvectors of $T_{\rho _{\psi }}^{\dagger }T_{\rho _{\psi }}$
corresponding to the eigenvalue $\mathrm{C}_{\rho _{\psi }}^{2}$ with
multiplicity $2$. Representation (\ref{26_2}) implies that the eigenvalues
of $G(\rho _{\psi })$ are equal to 
\begin{equation}
\eta _{1}=\frac{1+\left\Vert r_{2}\right\Vert _{\mathbb{R}^{3}}^{2}}{4}\text{%
, }\ \eta _{2,3}=\frac{\mathrm{C}_{\rho _{\psi }}^{2}}{4}.
\end{equation}%
For a maximally entangled two-qubit state $|\psi _{\max }\rangle ,$ the
Bloch vector $r_{2}=0,$  $T_{\rho _{\psi _{\max }}}^{\dagger }T_{\rho _{\psi
\max }}=\mathbb{I}_{\mathbb{R}^3}$ and $G(\rho _{\psi _{\max }})=\frac{1}{4}\mathbb{I}%
_{\mathbb{R}^3}.$ Thus, for any pure two-qubit state $|\psi \rangle ,$ by (\ref{25__})
we have%
\begin{equation}
\mathcal{D}_{g}(\rho _{\psi })=\sum_{n=2}^{3}\eta _{n}=\frac{1}{2}\mathrm{C}%
_{\rho _{\psi }}^{2}.  \label{26_3}
\end{equation}%
The value of the geometric quantum discord of a two-qubit state via its negativity $%
\mathcal{N}_{\rho _{\psi }}$ follows from (\ref{26_3}) and relation $\mathrm{C}_{\rho _{\psi
}}=2\mathcal{N}_{\rho _{\psi }}$ valid for every pure two-qubit state.

Let $d>2.$ Recall that, for any pure two-qudit state $|\psi
\rangle \langle \psi |$ on $\mathcal{H}_{d}\otimes \mathcal{H}_{d}$, the
non-zero eigenvalues $0<\mu _{k}(\psi )\leq 1$ of its reduced states
coincide and have the same multiplicity and vector $|\psi \rangle \in 
\mathcal{H}_{d}\otimes \mathcal{H}_{d}$ admits the Schmidt decomposition
\begin{equation}
|\psi \rangle  =\sum_{1\leq \text{ }n\text{ }\leq r_{sch}^{(\psi )}}\sqrt{%
\mu _{n}(\psi )\text{ }}|e_{n}^{(1)}\rangle \otimes |e_{n}^{(2)}\rangle ,
\sum_{1\leq n\leq r_{sch}^{(\psi )}}\mu _{n}(\psi )=1,
\label{01} \\
\end{equation}
where $\mu _{1}\geq \mu _{2}\geq ...\geq \mu _{r_{sch}^{(\psi )}}>0$ are nonzero eigenvalues of the reduced states of $\rho_{\psi}$, listed in the decreasing order and according to their multiplicity, and $
\lbrace|e_{k}^{(j)}\rangle \in \mathcal{H}\rbrace$ ,  $j=1,2,$ are sets of the corresponding mutually orthogonal unit  eigenvectors
of the reduced states. Parameters $\sqrt{\mu _{n}(\psi )%
\text{ }}$and $1\leq r_{sch}^{(\psi )}\leq d$ are called the Schmidt
coefficients and the Schmidt rank of $|\psi \rangle ,$ respectively. For simplicity of further calculations, we present the decomposition (\ref{01}) in the form 
\begin{equation}
|\psi \rangle  =\sum_{1\leq \text{ }n\text{ }\leq d }\sqrt{%
\mu _{n}(\psi )\text{ }}|e_{n}^{(1)}\rangle \otimes |e_{n}^{(2)}\rangle\label{01} \\
\end{equation}
by adding into the sum the zero eigenvalues $\mu _{n}$ of the reduces states if $n>r_{sch}^{(\psi )}$.

As it
is underlined in Theorem 1, the eigenvalues $\eta _{n}$ of the positive
operator $G(\rho ),$ given by (\ref{26}), are independent on a choice of
tuples $\Upsilon _{d_{1}}$ and $\Upsilon _{d_{2}}$ in representation (\ref{6__}). Therefore, in case of a pure two-qudit
state $\rho _{\psi }$, for finding in 
expression (\ref{25__}) the sum $\sum_{n=d_{2}}^{d_{2}^{2}-1}\eta _{n}$ of the
eigenvalues of $G(\rho _{\psi }),$ we take on each of
Hilbert spaces in $\mathcal{H}_{d}\otimes \mathcal{H}_{d}$ the tuple $%
\Upsilon _{d}$ of operators, which are similar by their structure to the
generalized Gell-Mann operators presented by relations (4)-(6) in \cite{Lou:20} but
expressed not via the elements of the standard basis in $\mathbb{C}^{d}$
but via the elements the corresponding orthonormal basis $\lbrace|e_{k}^{(j)}\rangle \in \mathcal{H}_{d}\rbrace$, 
$j=1,2$, in (\ref{01}). 

Under this choice,
by relations (\ref{8__}), (\ref{9__}), (\ref{10__}) and (\ref{01}) we find (quite similarly as it is done in Section 4 of \cite{Lou:20}) that the matrix representation of the operator $G(\rho_{\psi})$ is block-diagonal with the eigenvalues $\eta_{n}$ for $n\geq d$ equal to $\mu _{k}\mu _{m}$, $1\leq k<m\leq d$, each with multiplicity 2. Therefore, in (\ref{25__})
\begin{equation}
\sum_{n=d}^{d^{2}-1}\eta _{n}=2\sum_{1\leq k<m\leq d}\mu _{k}\mu _{m}.\label{00}\\\
\end{equation}
This and the relation  
\begin{equation}
\mathrm{C}_{\rho _{\psi }}^{2}=2\left( 1-\mathrm{tr[}\rho
_{j}^{2}]\right)= 4\sum_{1\leq k<m\leq d}\mu _{k}\mu _{m},
\end{equation} 
following from (\ref{013}) and (\ref{01}), prove the equality in (\ref{26_1}). 
The upper bound in (\ref{26_1}) follows\footnote{For the expression of the negativity of a pure state via its Schmidt coefficients, see footnote 3.} from the relation $C_{\rho _{\psi }}^{2}\leq 4\left(\sum_{1\leq k<m\leq d}\sqrt{\mu _{k}\mu _{m}}\right)^2=4%
\mathcal{N}^{2}_{\rho _{\psi }}
$, valid for any two-qudit state.

For a maximally entangled pure two-qudit state $\rho_{\psi_{max}}$, the concurrence is equal to $\frac{2(d-1)}{d}$ and by (\ref{26_1}) this implies (\ref{26__2}). The latter relation follows also directly from (\ref{25__}), since,  for a maximally entangled two-qudit state $\rho_{\psi_{max}}$,  the Bloch
vector $r_{2}=0$, the correlation matrix is diagonal \cite{Lou:20} with all its singular values equal to $\frac{2}{d}$. Therefore, 
\begin{equation}
G(\rho_{\psi_{max}})=\frac{1}{4}T_{\rho_{\psi_{max}}}^{\dagger}T_{\rho_{\psi_{max}}},\text{ \ \ \ 
}\eta_{n}=\frac{1}{d^{2}}, 
\end{equation}  
and, in view of (\ref{25__}), this implies 
\begin{equation}
\mathcal{D}_{g}(\rho_{\psi_{max}})=\sum_{n=d}^{d^{2}-1}\eta_{n}=\frac{d(d-1)}{d^{2}}%
=\frac{d-1}{d},   \label{26_5}\\
\end{equation}
i.e., expression (\ref{26__2}).
\end{proof}

\section{Upper and lower bounds}

In this Section, based on the new result of Theorem 1, we introduce the
general upper and lower bounds valid for an arbitrary two-qudit state, also, specify
the upper bound in case of a separable two-qudit state.

\begin{theorem}
For an arbitrary two-qudit state $\rho,$ pure or mixed, on $\mathcal{H}%
_{d_{1}}\otimes\mathcal{H}_{d_{2}},$ $d_{1},d_{2}\geq2,$ its geometric
quantum discord admits the following new upper bounds%
\begin{align}
\mathcal{D}_{g}(\rho) & \leq\frac{d_{2}-1}{d_{2}}\left[ 1-\frac{\left\Vert
r_{2}\right\Vert ^{2}}{d_{2}+1}\right]  \label{25_1} \\
& \leq\frac{d_{2}-1}{d_{2}} ,   \label{25-2}
\end{align}
where the upper bound (\ref{25-2}) constitutes the geometric discord of the
maximally entangled two-qudit state if $\min\{d_{1},d_{2}\}=d_{2}.$
\end{theorem}

\begin{proof}
Taking into account that%
\begin{equation}
\mathrm{tr}[G(\rho)]=\sum_{n=1}^{d_{2}^{2}-1}\eta_{n}=\frac{d_{2}-1}{%
d_{2}d_{1}}\left\Vert r_{2}\right\Vert _{\mathbb{R}^{d_{2}^{2}-1}}^{2}+\frac{%
1}{4}\mathrm{tr}[T_{\rho}^{\dagger}T_{\rho}].   \label{25-3}
\end{equation}
and relation (\ref{10.1}), we have%
\begin{equation}
\sum_{n=1}^{d_{2}^{2}-1}\eta_{n}\leq\frac{d_{1}d_{2}-1}{d_{1}d_{2}}-\frac{%
d_{1}-1}{d_{1}d_{2}}\left\Vert r_{1}\right\Vert _{\mathbb{R}%
^{d_{1}^{2}-1}}^{2}.   \label{25-4}
\end{equation}
We further note that since $\sum_{n=1}^{d_{2}^{2}-1}\eta_{n}=%
\sum_{n=1}^{d_{2}-1}\eta_{n}+$ $\sum_{n=d_{2}}^{d_{2}^{2}-1}\eta_{n}$ and $%
\sum _{n=d_{2}}^{d_{2}^{2}-1}\eta_{n}\leq d_{2}\sum_{n=1}^{d_{2}-1}\eta_{n}$%
, relation (\ref{25-4}) implies 
\begin{align}
\sum_{n=d_{2}}^{d_{2}^{2}-1}\eta_{n} & \leq\frac{1}{d_{2}+1}\left\{ \frac{%
d_{1}d_{2}-1}{d_{1}}-\frac{d_{1}-1}{d_{1}}\left\Vert r_{1}\right\Vert _{%
\mathbb{R}^{d_{1}^{2}-1}}^{2}\right\}  \label{25_4} \\
& =\frac{d_{1}d_{2}-1}{d_{1}(d_{2}+1)}-\frac{1}{d_{2}+1}\cdot\frac{d_{1}-1}{%
d_{1}}\left\Vert r_{1}\right\Vert ^{2}.  \notag
\end{align}
By using in (\ref{25_4}) equality (\ref{07_}), we derive%
\begin{align}
\sum_{n=d_{2}}^{d_{2}^{2}-1}\eta_{n} & \leq\frac{d_{1}d_{2}-1}{d_{1}(d_{2}+1)%
}-\frac{1}{d_{2}+1}\left( \frac{1}{d_{2}}-\frac{1}{d_{1}}+\frac{d_{2}-1}{%
d_{2}}\left\Vert r_{2}\right\Vert ^{2}\right)  \label{25_5} \\
& =\frac{d_{2}-1}{d_{2}}\left[ 1-\frac{\left\Vert r_{2}\right\Vert ^{2}}{%
d_{2}+1}\right] .  \notag
\end{align}
The latter also implies the upper bound (\ref{25-2}) and proves the
statement.
\end{proof}

Furthermore, based on the exact relation (\ref{25__}) we find the following
general upper and lower bounds on the geometric quantum discord of a
two-qudit state.

\begin{theorem}
For a two-qudit state $\rho$ on $\mathcal{H}_{d_{1}}\otimes\mathcal{H}%
_{d_{2} },$ $d_{1},d_{2}\geq2,$ the geometric quantum discord admits the
following new bounds: {\small 
\begin{align}
& \frac{1}{4}\mathrm{tr}[T_{\rho}^{\dagger}T_{\rho}]-\frac{1}{4}%
\sum_{n=1}^{d_{2}-1}\lambda_{n}  \label{27} \\
& \leq\mathcal{D}_{g}(\rho)  \notag \\
& \leq\min\left\{ \frac{1}{4}\mathrm{tr}[T_{\rho}^{\dagger}T_{\rho}];\text{ }%
\frac{1}{4}\mathrm{tr}[T_{\rho}^{\dagger}T_{\rho}]+\frac{d_{2}-1}{d_{2}d_{1}}%
\left\Vert r_{2}\right\Vert _{\mathbb{R}^{d_{2}^{2}-1}}^{2}-\frac{1}{4}%
\sum_{n=1}^{d_{2}-1}\lambda_{n}\right\} ,  \notag
\end{align}
}{\normalsize where $\lambda_{1}\geq\lambda_{2}\geq\cdots\lambda_{d^{2}-1}%
\geq0$ are eigenvalues of the positive Hermitian operator $%
T_{\rho}^{\dagger}T_{\rho}$ on $\mathbb{R}^{d_{2}^{2}-1}.$ }
\end{theorem}

\begin{proof}
For the evaluation of the last term in \eqref{17__}, we note that $%
\max_{x}\left\{ f_{1}(x)+f_{2}(x)\right\} \leq
\max_{x}f_{1}(x)+\max_{x}f_{2}(x),$ and, if $f_{j}(x)\geq 0,$ $j=1,2,$ then $%
\max_{x}\left\{ f_{1}(x)+f_{2}(x)\right\} \geq \max_{x}f_{j}(x),$ $j=1,2.$
These relations and Propositions 2 and 4 imply 
\begin{align}
& \max_{\Omega _{\Upsilon _{d_{2}}}}\mathrm{tr}\left[ \left( \frac{d_{2}-1}{%
d_{1}d_{2}}|r_{2}\rangle \langle r_{2}|+\frac{1}{4}T_{\rho }^{\dagger
}T_{\rho }\right) \Pi _{\Omega _{\Upsilon _{d_{2}}}}\right]   \label{27-1} \\
& \leq \max_{\Omega _{\Upsilon _{d_{2}}}}\mathrm{tr}\left[ \frac{d_{2}-1}{%
d_{1}d_{2}}|r_{2}\rangle \langle r_{2}|\Pi _{\Omega _{\Upsilon _{d_{2}}}}%
\right] +\max_{\Omega _{\Upsilon _{d_{2}}}}\mathrm{tr}\left[ \frac{1}{4}%
T_{\rho }^{\dagger }T_{\rho }\Pi _{\Omega _{\Upsilon _{d_{2}}}}\right]  
\notag \\
& =\frac{d_{2}-1}{d_{1}d_{2}}\left\Vert r_{2}\right\Vert _{\mathbb{R}%
^{d_{2}^{2}-1}}^{2}+\frac{1}{4}\sum_{n=1}^{d_{2}-1}\lambda _{n},  \notag
\end{align}%
as well as 
\begin{align}
& \max_{\Omega _{\Upsilon _{d_{2}}}}\mathrm{tr}\left[ \left( \frac{d_{2}-1}{%
d_{1}d_{2}}|r_{2}\rangle \langle r_{2}|+\frac{1}{4}T_{\rho }^{\dagger
}T_{\rho }\right) \Pi _{\Omega _{\Upsilon _{d_{2}}}}\right]   \label{27-2} \\
& \geq \max_{\Omega _{\Upsilon _{d_{2}}}}\mathrm{tr}\left[ \left( \frac{%
d_{2}-1}{d_{1}d_{2}}|r_{2}\rangle \langle r_{2}|\right) \Pi _{\Omega
_{\Upsilon _{d_{2}}}}\right] =\frac{d_{2}-1}{d_{1}d_{2}}\left\Vert
r_{2}\right\Vert _{\mathbb{R}^{d_{2}^{2}-1}}^{2}  \notag
\end{align}%
and 
\begin{align}
& \max_{\Omega_{\Upsilon_{d_{2}}}}\mathrm{tr}\left[ \left( \frac{d_{2}-1}{d_{1}d_{2}}%
|r_{2}\rangle \langle r_{2}|+\frac{1}{4}T_{\rho }^{\dagger }T_{\rho }\right)
\Pi _{\Omega_{\Upsilon_{d_{2}}}}\right]   \label{27-3} \\
& \geq \max_{\Omega_{\Upsilon_{d_{2}}}}\mathrm{tr}\left[ \frac{1}{4}T_{\rho }^{\dagger
}T_{\rho }\Pi _{\Omega_{\Upsilon_{d_{2}}}}\right] =\frac{1}{4}\sum_{n=1}^{d_{2}-1}\lambda
_{n}\ .  \notag
\end{align}%
Relations (\ref{27-1}), (\ref{27-2}) and (\ref{27-3}) prove the statement.
\end{proof}

From relation (\ref{10.1}) and the upper bound in (\ref{27}) it follows
that, for a two qudit state $\rho$ on $\mathcal{H}_{d_{1}}\otimes\mathcal{H}%
_{d_{2}},$ $d_{1},d_{2}\geq2,$%
\begin{equation}
\mathcal{D}_{g}(\rho)\leq\min\{J_{1},J_{2}\},   \label{28}
\end{equation}
where 
\begin{equation}
J_{1}=\frac{d_{1}d_{2}-1}{d_{2}d_{1}}-\frac{d_{1}-1}{d_{2}d_{1}}\left\Vert
r_{1}\right\Vert _{\mathbb{R}^{d_{1}^{2}-1}}^{2}-\frac{d_{2}-1}{d_{2}d_{1}}%
\left\Vert r_{2}\right\Vert _{\mathbb{R}^{d_{2}^{2}-1}}^{2}   \label{28-1}
\end{equation}
and 
\begin{equation}
J_{2}=\frac{d_{1}d_{2}-1}{d_{2}d_{1}}-\frac{d_{1}-1}{d_{2}d_{1}}\left\Vert
r_{1}\right\Vert _{\mathbb{R}^{d_{1}^{2}-1}}^{2}-\frac{1}{4}%
\sum_{n=1}^{d_{2}-1}\lambda_{n}.   \label{28-2}
\end{equation}
The new upper bounds (\ref{25_1}), (\ref{25-2}) in Theorem 2 and the new
upper bounds (\ref{27}), (\ref{28})-(\ref{28-2}) considerably improve the
upper bound in Proposition 3.1 of \cite{Bag.Dey.Osa:19} having in our
notations the form $\frac{d_{1}d_{2}-1}{d_{2}d_{1}}$.

Consider also the upper bound on the geometric quantum discord in a general
separable case.

\begin{proposition}
For every separable two-qudit state 
\begin{equation}
\rho _{sep}=\sum_{k}\beta _{k}\rho _{1}^{(k)}\otimes \rho _{2}^{(k)},\text{
\ \ }\beta _{k}\mathcal{\ \geq }0,\text{ \ \ }\sum_{k}\beta _{k}=1,
\label{30-0}
\end{equation}%
on $\mathcal{H}_{d}\otimes \mathcal{H}_{d},$ $d\geq 2,$ 
\begin{equation}
\mathcal{D}_{g}(\rho _{sep})\leq \left( \frac{d-1}{d}\right) ^{2}.
\label{30}
\end{equation}%
\end{proposition}

\begin{proof}
In view of expression (\ref{11-3}) for the correlation matrix $T_{\rho_{sep}}
$ of separable state $\rho_{sep},$ we have in case $d_{1}=d_{2}:$ 
\begin{align}
&\frac{1}{4}\mathrm{tr}[T_{\rho_{sep}}^{\dagger}T_{\rho_{sep}}]  \notag \\
&=\left( \frac{d-1}{d}\right) ^{2}\
\sum_{k,k_{1}}\beta_{k}\beta_{k_{1}}(r_{1}^{(k)}\cdot
r_{1}^{(k_{1})})(r_{2}^{(k)}\cdot r_{2}^{(k_{1})})  \label{30.1} \\
& \leq\left( \frac{d-1}{d}\right) ^{2},  \notag
\end{align}
where $r_{j}^{(k)}$ are the Bloch vectors of states $\rho_{j}^{(k)},$ $j=1,2,
$ given by (\ref{2_1}) with norms $||r_{1}^{(k)}||,$ $||r_{2}^{(k)}||$ $%
\leq1.$ 
\end{proof}

From the upper bound (\ref{30}) it follows that, for any separable two-qubit
state, the geometric quantum discord cannot exceed $1/4.$

If a separable state is pure, then $||r_{1}||=$ $||r_{2}||$ $=1$ and by (\ref{11-3}) the positive Hermitian operator (\ref{26}) takes the form%
\begin{align}
& G(\rho_{sep}^{(pure)}) =\frac{d-1}{d^{2}}|r_{2}\rangle\langle r_{2}|\text{ 
}+\text{ }\left( \frac{d-1}{d}\right) ^{2}\ \left\Vert r_{1}\right\Vert
^{2}|r_{2}\rangle\langle r_{2}|  \label{30.2} \\
& =\left( \frac{d-1}{d^{2}}+\left( \frac{d-1}{d}\right) ^{2}\right) \text{ }%
|r_{2}\rangle\langle r_{2}|=\frac{d-1}{d}|r_{2}\rangle\langle r_{2}|  \notag
\end{align}
and has only one nonzero eigenvalue 
\begin{equation}
\eta_{1}=\frac{d-1}{d}=\mathrm{tr}\left[ G(\rho_{sep}^{(pure)})\right] 
\label{30.3}
\end{equation}
with multiplicity one. Therefore, for a pure separable state, relation (\ref{25__}) gives $\mathcal{D}_{g}(\rho
_{sep}^{(pure)})=0$ 
-- as it should be since a pure separable state is quantum-classical. This is also consistent with (\ref{26_1}) since for a pure separable state the concurrence is equal to zero.

\section{Conclusion}
In the present article, for an arbitrary two-qudit state with any dimensions
at two sites, we find (Theorem 1) in the explicit analytical form the exact
value of its geometric quantum discord. This new rigorously proved general
result indicates that the lower bound on the geometric quantum discord found in \cite{Ran.Par:12} constitutes its exact value for each two-qudit state and includes only as particular cases the exact
results for a general two-qubit state \cite{Dak.Ved.Bru:10}, a general
qubit-qudit state \cite{Ran.Par:12} and some special families of two-qudit
states \cite{Ran.Par:12,Luo.Fu:10}.

Based on this new general result (\ref{25__}) of Theorem 1, we:

\begin{itemize}
\item[(a)] show (Corollary 1) that, for every pure two-qudit state, the exact
value of the geometric quantum discord is equal to one half of its squared
concurrence;

\item[(b)] find (Theorem 2) new general upper bounds (\ref{25_1}) and (\ref%
{25-2}) on the geometric quantum discord of an arbitrary two-qudit quantum
state of any dimension which are consistent with the exact value in
Corollary 1  for the geometric quantum discord of a maximally entangled pure
two-qudit state;

\item[(c)] derive (Theorem 3) for an arbitrary two-qudit state the new
general upper and lower bounds on the geometric quantum discord expressed
via the eigenvalues of its correlation matrix. These upper bounds are
tighter than the ones in \cite{Bag.Dey.Osa:19};

\item[(d)] specify (Proposition 3)the new upper bound on the geometric quantum discord of
an arbitrary separable two-qudit state of any dimension.
\end{itemize}

The new general results derived in the present article considerably extend
the range of known results on properties of the geometric quantum discord of
a two-qudit state, pure or mixed, of an arbitrary dimension.

As shown in \cite{Maz.etal:09, Maz.Pii.Man:10}, there exist bipartite quantum states which, under evolution via some quantum channels, exhibit a particular
type of decoherence with the following dynamics of correlations: until some critical value of time only classical correlations are being
destroyed while a decrease of quantum
correlations starts only after this critical time and this decrease is quantified via the quantum discord. This phenomenon, referred to as
the sudden transition of quantum correlations, occurs even in situations where
entanglement is monotonically decreasing since the initial time.

Even if under decoherence scenarios the geometric quantum discord could be more fragile than the quantum discord, as it is 
exemplified for diverse channels in \cite
{Ram:13}, this measure of quantum correlations is a useful concept to analyze quantum correlations dynamics and by using this measure phenomena like the sudden transition have been observed for some three and six-qubit GHZ states \cite{Ram:13}.
Similar studies for other N-qubit states were explored recently in \cite%
{Zhu.etal:22} and also in \cite{Hua.Qiu:16}.

The latter investigations \cite{Ram:13,Zhu.etal:22, Hua.Qiu:16} indicate that our explicit exact analytical expression (\ref{25__}) for the geometric
quantum discord could be a fundamental tool to study under diverse decoherence the time evolution of quantum correlations
in a general two-qudit system, where, to our knowledge, this measure has not been
explored. A possible other application of the geometric quantum discord is to quantify usefulness of a given state for teleportation tasks as long as it interacts
with the environment \cite{Adi.Ban:12}. 

We also expect that our new results will be of relevance in the growing field of relativistic quantum information. We can see the first steps on this direction by recent applications \cite{Ban.Bas.etal:23, Fan.Wen.etal:24} of geometric quantum discord for quantifying quantum correlations in quantum gravity contexts. 

With these investigations \cite{Ram:13,Zhu.etal:22, Hua.Qiu:16,Adi.Ban:12} in mind, we believe that our new results on the geometric quantum discord for an arbitrary  two-qudit state are important not only from the general theoretical point of view but also for the use of this measure of quantum correlations  in many practical tasks involving  two-qudit systems. 

\section*{Acknowledgment}

The authors acknowledge the valuable comments of the anonymous reviewers. The study was implemented in the framework of the Basic Research Program at
the HSE University.

\section{Appendix A}

\textbf{The proof of Lemma 1}. In view of relations (\ref{15__}), we have
\begin{align}
\Pi _{\Omega _{\Upsilon _{d}}}^{2} &=\left( \frac{d-1}{d}\right)
^{2}\sum_{k,l\in \{1,\dots ,d\}}|y_{k}\rangle \langle y_{k}|y_{l}\rangle
\langle y_{l}|\text{ }  \tag{A1}  \label{A1} \\
& =\left( \frac{d-1}{d}\right) ^{2}\sum_{k=1,..,d}|y_{k}\rangle \langle
y_{k}|\text{ }-\frac{d-1}{d^{2}}\sum_{k\neq l\in \{1,\dots
,d\}}|y_{k}\rangle \langle y_{l}|\text{ }  \notag 
\\
&=\left( \frac{d-1}{d}\right) ^{2}\sum_{k=1,..,d}|y_{k}\rangle \langle
y_{k}|\text{ }  \notag \\
& 
-\frac{d-1}{d^{2}}\sum_{k,l\in \{1,\dots ,d\}}|y_{k}\rangle \langle y_{l}|%
\text{ }+\frac{d-1}{d^{2}}\sum_{k=1,..,d}|y_{k}\rangle \langle y_{k}|  \notag 
\end{align}
\begin{equation}
=\frac{d-1}{d}\sum_{k=1,..,d}|y_{k}\rangle \langle y_{k}|\text{ }=\Pi_{\Omega _{\Upsilon _{d}}}. \ \ \  \phantom{..........................}  \notag \\
\end{equation}
Since

\begin{equation}
\mathrm{tr}[\Pi _{\Omega _{\Upsilon _{d}}}]=d-1,  \tag{A2}
\label{A2}
\end{equation}
the rank of every projection $\Pi _{\Omega _{\Upsilon _{d}}}$ is equal to $\left( d-1\right)$. 

Note that an orthogonal
projection $\Pi _{\Omega _{\Upsilon _{d}}}$ has eigenvalue $1$ of
multiplicity $(d-1)$ and eigenvalue $0$ of multiplicity $d(d-1)$. Relations (\ref{A1}) and (\ref{A2}) prove Lemma 1. $\blacksquare$

The proof of Theorem 1 implies the following new  general statement, which we use for finding the bounds in  Theorem 3.

\begin{proposition}
For an arbitrary positive Hermitian operator $A$ on $\mathbb{R}^{d^{2}-1}$ and the orthogonal projections
\begin{equation}
\Pi _{\Omega _{y}}=\frac{d-1}{d}\sum_{k=1}^{d}|y_{k}\rangle \langle y_{k}|\ ,
\label{A3}
\tag{A3}
\end{equation}
on $\mathbb{R}^{d^{2}-1}$ of rank $d-1$, which are specified in Lemma 1, the maximum
\begin{equation}
\max_{\Omega _{y}}\mathrm{tr}[A\Pi _{\Omega _{y}}]=\zeta _{1}+\cdots +\zeta
_{d-1},  \label{A4} \tag{A4}
\end{equation}%
where $\zeta _{1}\geq \zeta _{2}\geq \cdots \geq \zeta _{d^{2}-1}\geq 0$ are
the eigenvalues\ of $A$ listed in the decreasing order with the
corresponding algebraic multiplicities.
\end{proposition}

\end{document}